\documentclass[submission,copyright]{eptcs}
\providecommand{\event}{GCM 2022}

\title{A Graph-Transformational Approach for Proving the Correctness of Reductions between NP-Problems}
\author{Hans-J\"org Kreowski, Sabine Kuske, Aaron Lye, Aljoscha Windhorst
\institute{University of Bremen, Department of Computer Science\\ 
P.O.Box 33 04 40, 28334 Bremen, Germany}
\email{$\{$kreo,kuske,lye,windhorst$\}$@uni-bremen.de}
}
\def\titlerunning{\hspace{.3em}A Graph-Transformational Approach for Proving the Correctness of Reductions between NP-Problems}
\def\authorrunning{H.-J. Kreowski, S. Kuske, A. Lye \& A. Windhorst}

\usepackage{csquotes}
\usepackage{graphicx}
\usepackage{subcaption}

\usepackage{amsfonts,latexsym,amssymb,amsmath}
\usepackage{setspace}
\usepackage{multirow}
\usepackage{tikz}
\usetikzlibrary{patterns,%
								positioning,%
								calc,%
								scopes,%
								decorations.fractals,%
								shapes.geometric,%
								shadows,%
								decorations.pathreplacing,%
								decorations.pathmorphing,%
								decorations.markings,%
								decorations,%
								shapes,%
								arrows,
       automata,
       petri,
								snakes,
       lindenmayersystems,
       cd
}

\tikzstyle{every picture} = [every state/.style = {circle,draw,minimum size = 0.15cm, inner sep = 1pt}, every edge/.style = {draw,inner sep = 2pt},on grid,inner sep = 0.03cm, every loop/.style={}, node distance = 0.4cm]
\tikzset{
    partial ellipse/.style args={#1:#2:#3}{
        insert path={+ (#1:#3) arc (#1:#2:#3)}
    }
}
\tikzset{
    wavy/.style = {->,very thin,double distance = 0,shorten <=.2pt,>=stealth',snake=snake, segment amplitude=1pt, segment length=8pt, line after snake=4pt}
}
\tikzset{
    wavybd/.style = {<->,very thin,double distance = 0,shorten <=.2pt,>=stealth',snake=snake, segment amplitude=1pt, segment length=8pt, line before snake=4pt, line after snake=4pt}
}

\tikzstyle{every picture} = [every state/.style = {circle,draw,minimum size = 0.15cm, inner sep = 1pt}, every edge/.style = {draw,inner sep = 2pt},on grid,inner sep = 0.03cm, every loop/.style={}, node distance = 0.4cm]
\usepgflibrary{shapes.geometric}

\usepackage{xifthen}
\usepackage{xspace}
\usepackage{latexsym,amssymb,amsmath,amsthm}
\usepackage{hyperref}
\usepackage{tabularx}
\usepackage{enumitem}

\usepackage{changes}

\newcommand{\dder}{\mathop{\Longrightarrow}\limits}

\newcommand{\rdder}{\mathop{\Longleftarrow}\limits}

\newcommand{\F}{\mathcal{F}}
\newcommand{\G}{\mathcal{G}}

\newcommand{\Nat}{\mathbb{N}}

\newcommand{\cardinality}[1]{\lvert #1 \rvert}

\newcommand{\iso}{\cong}

\newcommand{\NP}{\mathrm{NP}}

\newcommand{\DEC}{\mathit{DEC}}
\newcommand{\BOOL}{\mathit{BOOL}}

\newcommand{\TRUE}{\mathit{TRUE}}
\newcommand{\FALSE}{\mathit{FALSE}}

\newcommand{\pol}{\mathit{p}}
\newcommand{\size}{\mathit{size}}

\newcommand{\forbidden}{\mathit{forbidden}}
\newcommand{\standard}{\mathit{standard}}

\newcommand{\SEM}{\mathit{SEM}}

\newcommand{\DS}{\mathit{DS}}
\newcommand{\CDS}{\mathcal{DS}}
\newcommand{\RED}{\mathit{RED}}
\newcommand{\red}{\mathit{red}}

\newcommand{\conflux}{\mathit{conflux}}
\newcommand{\interchange}{\mathit{interchange}}
\newcommand{\sprout}{\mathit{sprout}}
\newcommand{\couple}{\mathit{couple}}
\newcommand{\associate}{\mathit{associate}}
\newcommand{\funct}{\mathit{funct}}
\newcommand{\gtu}{\mathit{gtu}}

\newcommand{\AS}{\mathit{AS}}

\newcommand{\rl}[1]{rule $#1$}

\newcommand{\lb}[1]{$\mathit{#1}$}

\newcommand{\lp}[1]{\lb{#1}-loop}

\newcommand{\hamcycle}{\mathit{hamcycle}\xspace}

\newcommand{\TSP}{\mathit{TSP}\xspace}
\newcommand{\hpmark}{\alpha}
\newcommand{\hampath}{\mathit{hampath}\xspace}
\newcommand{\clique}{\mathit{clique}\xspace}
\newcommand{\independentset}{\mathit{independent\text{-}set}\xspace}
\newcommand{\cliquetoindependentset}{\mathit{clique\text{-}to\text{-}independent\text{-}set}\xspace}
\newcommand{\hampathtostwbdtwo}{\mathit{hampath\text{-}to\text{-}2\text{-}bounded\text{-}spantree}\xspace}
\newcommand{\stwbd}[1]{\mathit{#1\text{-}bounded\text{-}spantree}\xspace}
\newcommand{\bound}[1]{\mathit{bound}(#1)\xspace}

\theoremstyle{plain}
\newtheorem{definition}{Definition}
\newtheorem{fact}{Fact}
\newtheorem*{proofprocedure}{Proof procedure}
\theoremstyle{definition}
\newtheorem{example}{Example}
\newtheorem{observe}{Observation}

\usepackage{fontawesome}

\colorlet{ext}{gray}
\colorlet{annotation}{gray}

\begin{document}
\maketitle

\begin{abstract}
The complexity class $\NP$ of decision problems that can be solved nondeterministically in polynomial time is of great theoretical and practical importance where the notion of polynomial-time reductions between $\NP$-problems
is a key concept for the study of $\NP$.
As many typical $\NP$-problems are naturally described as graph problems, they and their reductions are obvious candidates to be investigated by graph-transformational means.
In this paper, we propose such a graph-transformational approach for proving the correctness of reductions between $\NP$-problems.

\end{abstract}

\section{Introduction}
\label{sec:introduction}
The complexity class $\NP$ of decision problems that can be solved nondeterministically in polynomial time is of great theoretical and practical importance, with the notion of polynomial-time reductions between $\NP$-problems being a key concept for the study of $\NP$ and, in particular, of $\NP$-completeness.
Since the early analysis of the class $\NP$ it is known that many graph-theoretic problems are $\NP$-complete (cf. e.g.~\cite{Karp:72,Garey-Johnson:79}). 
Examples include finding a clique, a Hamiltonian cycle, a vertex cover, an independent set, a vertex coloring, etc. 
Moreover, many typical $\NP$-problems are naturally described as graph problems, such as routing and scheduling problems as well as various further optimization and planning problems. 
Usually the complexity class $\NP$ as well as the notion of complexity-theoretic reductions are defined by means of (polynomial-time) Turing machines. 
However, explicit problems and constructions are described on some higher level with much more abstraction.
Polynomial graph transformation units \cite{Kreowski-Kuske-Rozenberg:2008,Kreowski-Kuske:2012} may be helpful not only for specifying and understanding decision problems and reductions, but also for obtaining correctness proofs in a systematic way. 
In~\cite{Kreowski-Kuske:2010c}, it has been shown that polynomial graph transformation units are a formal computational model for decision problems in $\NP$.
In~\cite{Ermler-Kuske-Luderer-vonTotth:2013}, the problems of finding a clique, an independent set, a vertex cover and a Hamiltonian cycle are modeled as graph transformation units, and reductions are characterized by deadlock-free and confluent polynomial graph transformation units.
Moreover, a proof principle is proposed for proving the correctness of reductions.
It is based on a certain interaction of the unit of an $\NP$ problem and the unit of a reduction requiring a family of auxiliary variants of the reduction unit.
In this paper, we continue this research by proposing a novel graph-transformation-based framework for proving the correctness of reductions between $\NP$-problems.
As correctness requires to construct certain derivations of one kind from certain derivations of another kind the initial graphs of which are connected by a third kind of derivations, we provide a toolbox of operations on derivations that allows such constructions.
In particular, we employ operations based on parallel and sequential independence of rule applications.

The paper is organized as follows.
In Section~\ref{sec:preliminaries} and \ref{sec:gtu}, we recall the notions of graph transformation and graph transformation units.
In Section~\ref{sec:modeling-problems}, we present the notion of graph-transformational modeling of NP-problems.
In Section~\ref{sec:modeling-reductions}, we introduce the notion of graph-transformational modeling of reductions.
Section~\ref{sec:correctness-reductions} discusses the correctness proofs of such reductions.
Section~\ref{sec:conclusion} concludes the paper.

\section{Preliminaries}
\label{sec:preliminaries}
In this section, we recall the basic notions and notations of graphs and rule-based graph transformation as far as needed in this paper.

\begin{paragraph}{Graphs}
Let $\Sigma$ be a set of labels with $* \in \Sigma$.
A (directed edge-labeled) \emph{graph} over $\Sigma$ is a system $G = (V,E,s,t,l)$ where $V$ is a finite  set of \emph{vertices}, $E$ is a finite set of \emph{edges}, $s,t\colon E\to V$ and $l\colon E\to\Sigma$ are mappings assigning a \emph{source}, a \emph{target} and a \emph{label} to every edge $e\in E$.
An edge $e$ with $s(e) = t(e)$ is  called a \emph{loop}.
An edge with label $*$ is called an \emph{unlabeled edge}.
In drawings, the label $*$ is omitted. 
\emph{Undirected edges} are pairs of edges between the same vertices in opposite directions.
The components $V$, $E$, $s$, $t$, and $l$ of $G$ are also denoted by $V_G$, $E_G$, $s_G$, $t_G$, and $l_G$, respectively.
The empty graph is denoted by $\emptyset$.
The class of all directed edge-labeled graphs is denoted by $\G_\Sigma$.

For graphs $G,H \in \G_{\Sigma}$, a \emph{graph morphism} $g\colon G\to H$ is a pair of mappings $g_V\colon V_G\to V_H$ and $g_E\colon E_G\to E_H$ that are structure-preserving, i.e.,
$g_V(s_G(e)) = s_H(g_E(e))$,
$g_V(t_G(e)) = t_H(g_E(e))$,
and $ l_G(e) = l_H(g_E(e))$ for all $e \in E_G$.
If the mappings $g_V$ and $g_E$ are bijective, then $G$ and $H$ are \emph{isomorphic}, denoted by $G \cong H$. 
If they are inclusions, then $G$ is called a \emph{subgraph} of $H,$ denoted by $G \subseteq H$.
For a graph morphism $g \colon G\to H$, the image of $G$ in $H$ is called a \emph{match} of $G$ in $H$, i.e., the match of $G$ with respect to the morphism $g$ is the subgraph $g(G) \subseteq H$. 
\end{paragraph}

\begin{paragraph}{Rules and rule application}
A \emph{rule}  $r = (  L\supseteq K\subseteq R)$ consists of three graphs $L,K,R \in \G_{\Sigma}$ such that $K$ is a subgraph of $L$ and $R$.
The components $L$, $K$, and $R$ are called \emph{left-hand side}, \emph{gluing graph}, and \emph{right-hand side}, respectively. 

The application of $r$ to a graph $G$ consists of the following three steps.
%\begin{enumerate}
%  \item
(1)~Choose a match $g(L)$ of $L$ in $G$.
 % \item
(2)~Remove the vertices of $g_V(V_L) \setminus g_V(V_K)$ and the edges of $g_E(E_L) \setminus g_E(E_K)$ yielding $Z$, i.e.,\linebreak $Z=G-(g(L)-g(K))$.
  %\item
(3)~Add $R$ to $Z$ by gluing $Z$ with $R$ in $g(K)$ (up to isomorphism) yielding $H$.
%\end{enumerate}
The construction is subject to the \emph{dangling condition} ensuring that $Z$ becomes a subgraph of $G$ so that $H$ becomes a graph automatically.
Moreover, we require the \emph{identification condition}, i.e., if different items of $L$ are mapped to the same item in $g(L)$, then they are items of $K$.
Sometimes any identification may be forbidden. Then we add the postfix \emph{(inj)} to the rule.
The construction produces a right match $h(R)$ that extends $g(K)$ by the identity on $R-K$ (up to isomorphism).

The application of $r$ to $G$ w.r.t. $g$ is called \emph{direct derivation} and is denoted by $G \dder_r H$
(where $g$ is kept implicit). 
A \emph{derivation} from $G$ to $H$ is a sequence of direct derivations $G_0\dder_{r_1}G_1\dder_{r_2} \cdots \dder_{r_n} G_n$ with $G_0 = G$, $G_n \iso H$ and $n \ge 0$.
$r_1 \cdots r_n$ is called \emph{application sequence}. If $r_1,\cdots, r_n \in P$, then the derivation is also denoted by $G\dder_P^n H$.
If the length of the derivation does not matter, we write $G \dder_P^*H$. %and if the underlying rule set is clear from the context, the subscript $P$ may be omitted.

A rule $r = (L \supseteq K \subseteq R)$ may be equipped with a \emph{negative application condition} (NAC) $N$ with $L \subseteq N$. It prevents the application of $r$ to a graph $G$ if the match $g(L)$ can be extended to a match $\overline{g}(N)$.

Let $\cardinality{S}$ denote the cardinality of a finite set $S$
and $size(G)$ = $\cardinality{V_G} + \cardinality{E_G}$ the size of a graph $G$.
It is worth noting that the application of a given (fixed) rule to $G$ can be performed in polynomial time provided that the equality of labels can be checked in polynomial time.
This is due to the fact that
the number of graph morphisms from $L$ to $G$ is bounded by $size(G)^{size(L)}$ so that a match can be found in polynomial time.
The further steps of the rule application can be done in linear time.

It may be noted that the chosen notion of rule application fits into the DPO framework as introduced in~\cite{Ehrig-Pfender-Schneider73a}
(see, e.g., \cite{Corradini-Ehrig-Loewe.ea96a} for a comprehensive survey).
For details on NACs, we refer to \cite{Habel-Heckel-Taentzer96a}.
\end{paragraph}

\begin{paragraph}{Extension and independence}
Let $d=(G \dder^* H)$ be a derivation and $\widehat{G}$ be a graph with $G \subseteq \widehat{G}$
obtained by adding some vertices and some edges with new or old vertices as sources and targets.
Then $d$ can be extended to $\widehat{d}=(\widehat{G} \dder^* \widehat{H})$ if
none of the direct derivations of $d$ removes a source or a target of a new edge.
To construct $\widehat{d}$, one adds the new vertices and edges to each graph of $d$.
As the chosen sources and targets of new edges are never removed within $d$,
they can be used throughout $\widehat{d}$. By construction, each graph of $d$ is
subgraph of the corresponding extended graph so that the matches are kept
intact and the rules can be applied as in $d$. The extension works also in case
of negative context conditions if the new vertices and edges do not contradict them.

Let $r_i = (L_i \supseteq K_i \subseteq R_i)$ for $i = 1,2$ be rules. 
Two direct derivations $G \dder_{r_i} H_i$ with matches $g_i(L_i)$ are \emph{parallel independent} if $g_1(L_1) \cap g_2(L_2) \subseteq g_1(K_1) \cap g_2(K_2)$.
Successive direct derivations $G \dder_{r_1} H_1 \dder_{r_2} X$ with the right match $h_1(R_1)$ and the (left) match $g'_2(L_2)$ are \emph{sequentially independent} if $h_1(R_1) \cap g'_2(L_2) \subseteq h_1(K_1) \cap g'_2(K_2)$.
It is well-known that parallel independence induces the direct derivations $H_1 \dder_{r_2} X$ and $H_2 \dder_{r_1} X$ with matches $g_2'(L_2) \iso g_2(L_2)$ and $g_1'(L_1) \iso g_1(L_1)$ respectively
and that sequential independence induces the derivation $G \dder_{r_2} H_2 \dder_{r_1} X$ with matches $g_2(L_2) \iso g'_2(L_2)$ and $g'_1(L_1) \iso g_1(L_1)$ (cf., e.g., \cite{Corradini-Ehrig-Loewe.ea96a,Ehrig-Kreowski76,Ehrig-Rosen:76}).

The last two constructions can be extended to derivations by simple inductions.
Sets of rules $P_1$ and $P_2$ are called \emph{independent} if each two applications of a rule of $P_1$ and a rule of $P_2$ applied to the same graph are parallel independent and, applied one after the other, sequentially independent. 
In the special case of $P = P_1 = P_2$, we say that $P$ is independent.
Then the following hold:
%\begin{enumerate}
%  \item
(1)~$G \dder^{n}_{P_1} H_1$ and $G \dder^{m}_{P_2} H_2$ induce $H_1 \dder^{m}_{P_2} X$ and $H_2 \dder^{n}_{P_1} X$ for some $X$, and
%  \item
(2)~$G \dder^{n}_{P_1} H_1 \dder^{m}_{P_2} X$ induces $G \dder^{m}_{P_2} H_2 \dder^{n}_{P_1} X$ for some $H_2$.
%\end{enumerate}
The derivations are obtained by repeating the respective constructions for direct derivations of the given derivations as long as possible.
One way to look at the situation is that the $P_2$-derivation is moved either forward or backward along the $P_1$-derivation.
Therefore, we refer to the resulting $P_2$-derivation as \emph{moved variant} of the given $P_2$-derivation. 
\end{paragraph}

\section{Graph Transformation Units}
\label{sec:gtu}
In this section, the basic notions and notations of graph transformation units (see, e.g., \cite{Kreowski-Kuske-Rozenberg:2008,Kreowski-Kuske00a}) are recalled using graphs, rules and rule application as introduced in the previous section. 
Besides a set of rules, a graph transformation unit provides two graph class expressions to specify initial and terminal graphs and a control condition to regulate the derivation process.

We restrict the consideration to graph class expressions $e$ that specify graph classes $\SEM(e) \subseteq \G_{\Sigma}$ such that membership can be decided in polynomial time. 
In examples, we use
(1)~$\forbidden(\F)$ for some finite $\F \subseteq \G_{\Sigma}$ with $\SEM(\forbidden(\F))$ $=$ $\{ G \mid$ there is no match $f\colon F \to G$ for any $F \in \F \}$ and
(2)~some constant terms like \emph{undirected}, \emph{unlabeled}, \emph{simple}, and \emph{loop-free} referring to the class of undirected graphs, unlabeled graphs, simple graphs, i.e., graphs without parallel edges, and loop-free graphs respectively.
The term $\mathit{standard}$ specifies the intersection of these four graph classes.
Moreover, we use the term $\bound{\Nat}$ that specifies the set of all graphs with a single vertex, a single \lb{bound}-loop and an arbitrary number of unlabeled loops.
Another useful graph class expression is $\mathit{reduced}(P)$ for some finite set of rules $P$ that specifies the class of graphs to which none of the rules of $P$ can be applied. Finally, we use the binary operator $+$ for graph class expressions
specifying the disjoint union of the graphs of the corresponding graph classes.
The membership problems of all explicitly used graph class expressions can be
checked in polynomial time. $\mathit{forbidden}$ and $\mathit{reduced}$ are variants of the
matching problem. The properties of the $\mathit{standard}$ expressions follow from an
inspection of the edges which applies also to $\bound{\Nat}$. The disjoint
union symbol is only used in $\mathit{standard}+\bound{\Nat}$. Given a graph, one can
identify the component $\bound{k}$ for some $k$ by inspecting all edges.
The rest of the graph must be checked with respect to $\standard$.

As control conditions, we use extended regular expressions over sets of rules.
Let $C$ be a regular expression over some set of rules $P$ and $d = ( G \dder^{*}_{P} H )$.
Then $d$ \emph{satisfies} $C$ if the application sequence of $d$ is in $L(C)$ where $L(C)$ is the regular language specified by $C$. This is denoted by $G \dder^{*}_{P,C} H$.
Employing a finite automaton corresponding to the regular expression, one can control such derivations stepwise, meaning that the allowed and applicable rules can be determined.
In this way, satisfaction can be checked in polynomial time for derivations of polynomial lengths.
In an \emph{extended regular expression}, the Kleene star $^{*}$ behind a rule may be replaced by $!$, where $r!$ requires that $r$ be applied as long as possible
(in contrast to $r^*$ that allows to apply $r$ arbitrarily often).
%Therefore, whether $r$ can be applied in the next step is part of the stepwise control of the regular expression.
As the stepwise control provides all currently applicable rules, the applicability of $r$ is included so that the control of $r!$ does not require any extra time.
In examples, the rules of graph transformation units are listed such that their position in the list ($1,2,3,\ldots$) can be used to refer to them in control conditions.

A \emph{graph transformation unit} is a system $\gtu = (I,P,C,T)$ where $I$ is an \emph{initial} graph class expression, $P$ is a finite set of rules, $C$ is an extended regular expression over $P$ and $T$ is a \emph{terminal} graph class expression.
The components $I$, $P$, $C$ and $T$ of $\gtu$ are also denoted by $I_{\gtu}$, $P_{\gtu}$, $C_{\gtu}$ and $T_{\gtu}$, respectively.
The \emph{semantics} of $\gtu$ is the binary relation $\SEM(\gtu) = (\SEM(I) \times \SEM(T)) \cap \dder^{*}_{P,C}$. 
A derivation $G \dder^{*}_{P,C} H$ with $G \in \SEM(I)$ and $H \in \SEM(T)$ is called \emph{successful}, denoted by $G \dder^{*}_{\gtu} H$.
As long as a derivation $G \dder^{*}_{P} H$ with $G \in \SEM(I)$ follows the stepwise control of $C$, it is called \emph{permitted} and denoted by $G \dder^{*}_{P,C?} H$.

In examples, a graph transformation unit is presented schematically where the 
components $I$, $P$, $C$, and $T$ are listed after respective keywords initial, rules, 
cond, and terminal.

A graph transformation unit is \emph{polynomial} if there is a polynomial $p$ such that $n \leq p(\size(G))$ for each permitted derivation $G \dder^{n}_{P,C?} H$.
A polynomial graph transformation unit $\gtu$ is \emph{functional} if,
for every initial graph $G$, the following holds:
(1)~There is a successful derivation $G \dder^* H$ for some graph $H$.
(2)~Every permitted derivation $G \dder^* \overline{H}$ can be prolonged into
a successful derivation $G \dder^* \overline{H} \dder^* H$.
(3)~If $G \dder^* H$ and $G \dder^* H'$ are successful derivations, then $H \iso H'$.
In this case, the resulting graph is denoted by $\gtu(G)$.
The first property guarantees that every initial graph gets a result.
The second property may be seen as a kind of deadlock freeness.
And the third property is equivalent to the confluence of permitted
derivations whereas arbitrary derivations may not be confluent. The three
properties of functional units allow for every initial graph $G$ to derive the
resulting graph $gtu(G)$ in polynomial time by starting with the 0-derivation of
$G$ and adding derivation steps that keep the permission as long as possible.

Note that we restrict ourselves to graph class expressions with polynomial membership problems for specifying initial and terminal graphs and to extended regular expressions as control conditions  that can be checked stepwise in polynomial time. Therefore, our definition of polynomial graph transformation units
is a special case of the one given in~\cite{Ermler-Kuske-Luderer-vonTotth:2013}
where a more general kind of control conditions is allowed.

\section{Graph-Transformational Modeling of NP-Problems}
\label{sec:modeling-problems}
Based on its input-output semantics, a graph transformation unit can also be interpreted as a model of a decision problem. 
This problem belongs to the class $\NP$ if the length of each derivation of the unit is polynomially bounded and the proper nondeterminism is provided by the general functioning of units. 
This means that the usual technique to prove polynomial termination can be applied to show that the decision problem of a unit is in $\NP$.
%The four examples are further considered in Sections~\ref{sec:modeling-reductions} and~\ref{sec:correctness-reductions}.

\begin{definition}
Let $\gtu = (I,P,C,T)$ be a graph transformation unit. 
Then the decision problem of $\gtu$, denoted by $\DEC(\gtu)\colon \SEM(I) \to \BOOL$, yields $\TRUE$ for $G \in \SEM(I)$ if $(G,H) \in \SEM(\gtu)$ for some $H \in \SEM(T)$ and $\FALSE$ otherwise.
\end{definition}

\begin{observe}
Let $\gtu$ be a graph transformation unit.
Then $\DEC(\gtu) \in \NP$ if $\gtu$ is polynomial.
\end{observe}

\begin{proof}
By definition, the membership problems of the classes of initial and terminal graphs, and the stepwise control of rule application are polynomial and the lengths of permitted derivations are polynomially bounded.
Moreover, the rule application needs polynomial time and its nondeterminism is polynomially bounded so that $\DEC(\gtu) \in \NP$.
\end{proof}

\begin{example}%[Hamiltonian Path Problem]
\label{example_hampath}
The Hamiltonian path problem asks whether a graph contains a simple path that visits all vertices.
The problem is modeled by the unit in Figure~\ref{gtu:hampath}.
The length of each of its derivations is obviously bounded by twice the number of vertices of the initial graph such that $\DEC(\hampath) \in \NP$.
A successful derivation for a concrete instance is depicted in Figure~\ref{fig:hampath-example-derivation}.
\begin{figure}[h!]
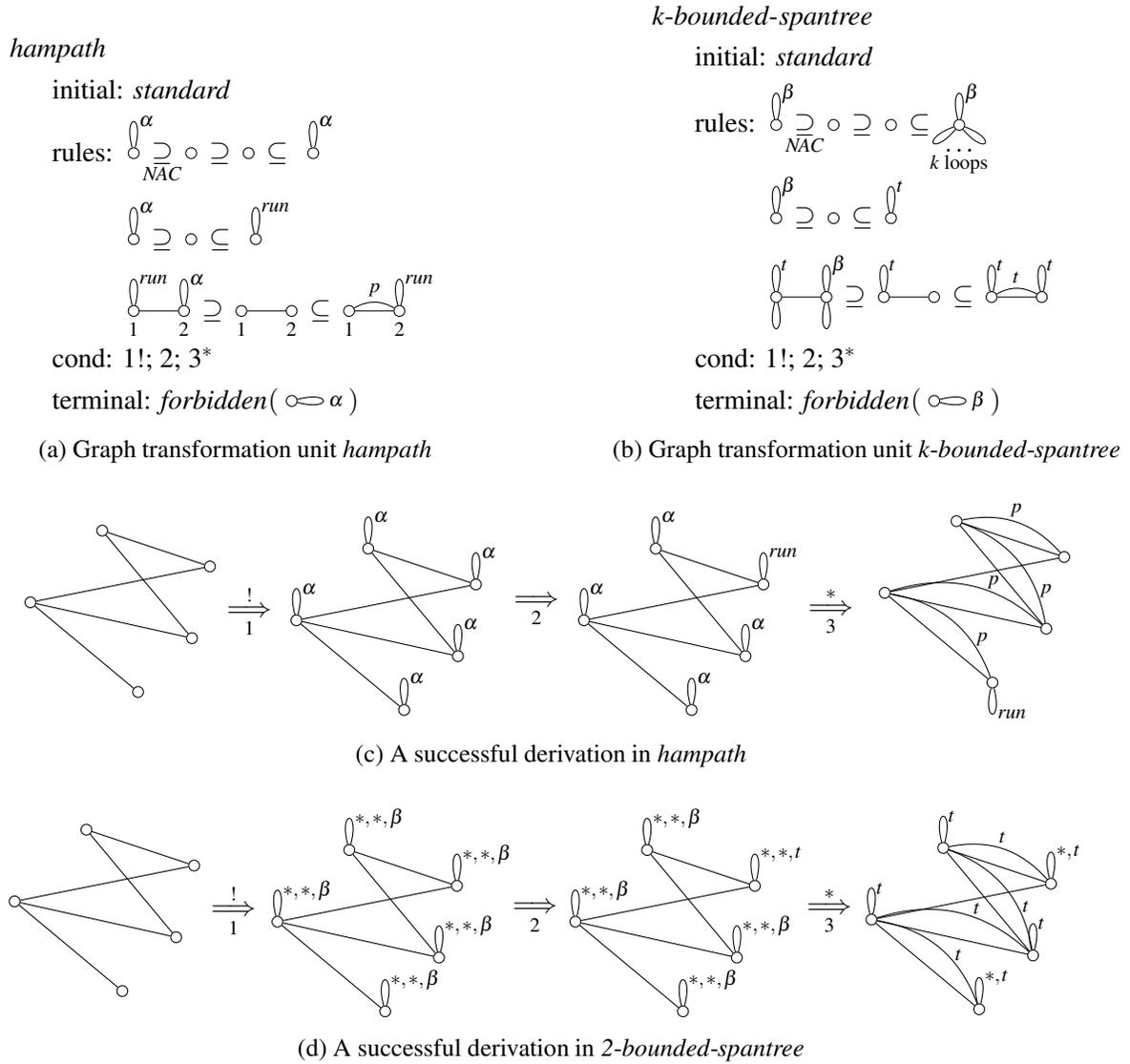

\centering
	\begin{subfigure}[b]{.45\textwidth}
		\centering
		\input{units/hampath}
		\caption{Graph transformation unit $\hampath$}
		\label{gtu:hampath}
	\end{subfigure}
	\hfill
	\begin{subfigure}[b]{.45\textwidth}
		\centering
		\input{units/stwbd}
		\caption{Graph transformation unit $\stwbd{k}$}
		\label{gtu:stwbd}
	\end{subfigure}
	\begin{subfigure}[b]{\textwidth}
		\centering
		\vspace{1em}
		\input{tikz/example-hampath-3-derivation}
		\caption{A successful derivation in $\hampath$}
		\label{fig:hampath-example-derivation}
	\end{subfigure}
	\begin{subfigure}[b]{\textwidth}
		\centering
		\vspace{1em}
		\input{tikz/example-2spantree-3-derivation}
		\caption{A successful derivation in $\stwbd{2}$}
		\label{fig:stwbd-example-derivation}
	\end{subfigure}
        \caption{Graph transformation units $\hampath$ and $\stwbd{k}$ and example derivations in the units}
	\label{gtu:hampath-and-stwbd}
\end{figure}
\end{example}

\begin{example}%[Spanning Tree with Bounded Degree Problem]
\label{example_stwbd}
The $k$-bounded spanning tree problem asks whether, for a given bound $k \in \Nat$, a graph contains a spanning tree with a vertex degree not greater than $k$.
The problem is modeled by the unit in Figure~\ref{gtu:stwbd}.
The length of each of its derivations is obviously bounded by twice the number of vertices of the initial graph such that $\DEC(\stwbd{k}) \in \NP$.
A successful derivation for a concrete instance for the case $k=2$ is depicted in Figure~\ref{fig:stwbd-example-derivation}. We use lists of labels for loops as a abreviation to clarify the drawings.
E.g., a loop with labels  $*,*,t$ represents two unlabeled loops and one $t$-labeled loop.
\end{example}

An often easy way to show the polynomiality of a unit $\gtu$ and as a consequence that $\DEC(\gtu) \in \NP$ is by finding a natural number variable the value of which is polynomially bounded and decreases whenever a rule of $\gtu$ is applied. 
This is a well-known fact.

\begin{fact}[Polynomial Termination]\label{fact_polynomial_termination}
\rm
Let $\gtu = (I,P,C,T)$ be a graph transformation unit, $f\colon \G_{\Sigma} \to \Nat$ be a function and $p$ a polynomial such that $f(x) \leq (\pol \circ \size)(x)$ for each $x \in \G_\Sigma$ and $f(G) > f(H)$ for each direct derivation $G \dder_{P} H$.
Then the length $k$ of every successful derivation $G \dder_{P}^{k} H$ is polynomially bounded in the size of the input graph $G$ such that $\DEC(gtu) \in \NP$.
\end{fact}

\begin{example}
The polynomial termination of $\hampath$ may be shown using Fact~\ref{fact_polynomial_termination} by separating the unit into two units. 
The first one with rule $1$ as its only rule terminates in a linear number of steps as the number of vertices without \lp{\alpha} decreases if a rule is applied. 
The second unit with rules $2$ and $3$ terminates in a linear number of steps as the number of vertices with \lp{\alpha}s decreases if a rule is applied. 
The polynomial termination of $\stwbd{k}$ can be shown analogously.
\end{example}

\begin{example}
The independent set problem asks whether a graph contains a set of $k$ vertices, no two of which are connected by an edge.
The problem is modeled by the unit in Figure~\ref{gtu:independentset}.
As the number of loops decreases whenever the rule is applied, one gets $\DEC(\independentset) \in \NP$.
A successful derivation for a concrete instance is depicted in Figure~\ref{fig:independentset-example-derivation}.
\end{example}
\begin{figure}[h]
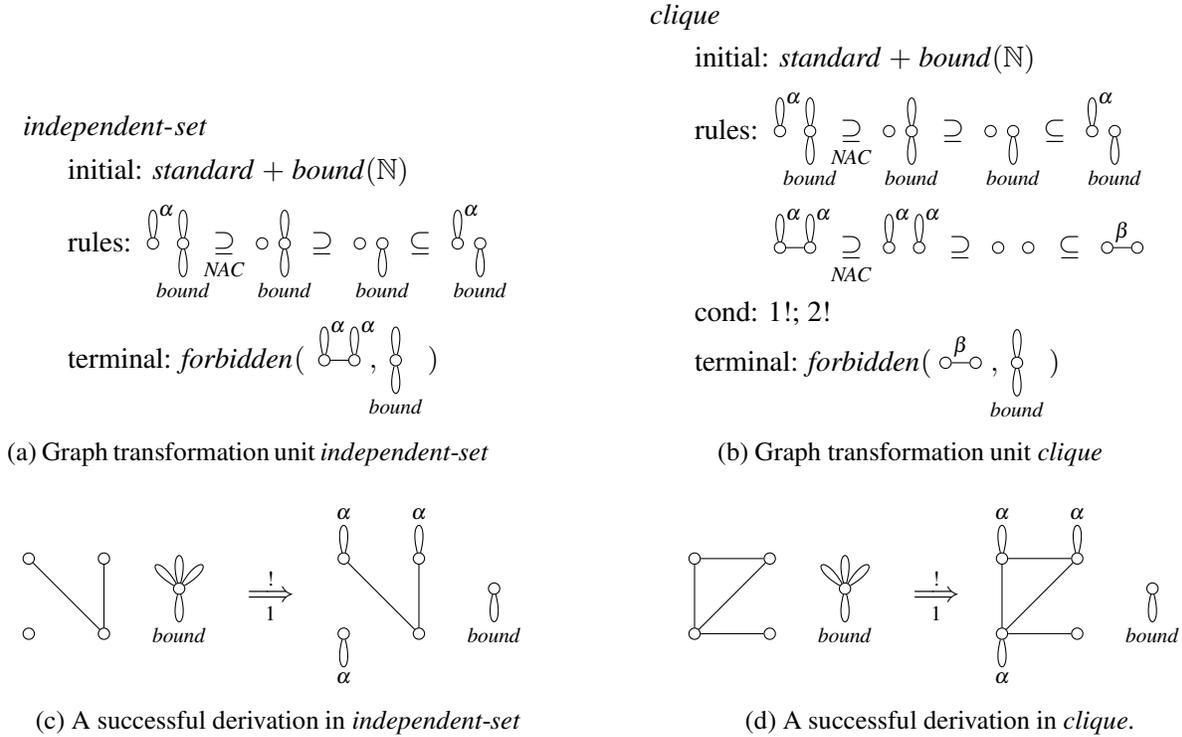

\centering
	\begin{subfigure}[b]{.4\textwidth}
		\centering
		\input{units/independentset}
		\caption{Graph transformation unit $\independentset$}
		\label{gtu:independentset}
	\end{subfigure}
	\hfill
	\begin{subfigure}[b]{.5\textwidth}
		\centering
		\input{units/clique}
		\caption{Graph transformation unit $\clique$}
		\label{gtu:clique}
	\end{subfigure}
	\begin{subfigure}[b]{0.45\textwidth}
		\centering
		\vspace{1em}
		\input{tikz/example-independentset-1-derivation}
		\caption{A successful derivation in $\independentset$}
		\label{fig:independentset-example-derivation}
	\end{subfigure}
	\hfill
	\begin{subfigure}[b]{0.45\textwidth}
		\centering
		\vspace{1em}
        \input{tikz/example-clique-1-derivation}
		\caption{A successful derivation in $\clique$.}
		\label{fig:clique-example-derivation}
	\end{subfigure}
        \caption{Graph transformation units $\independentset$ and $\clique$ and successful derivations in the units}	
	\label{gtu:independentset-and-clique}
\end{figure}

\begin{example}
The clique problem asks whether a graph contains a set of $k$ vertices every two of which are connected by an edge, constituting a \emph{complete} subgraph.
The problem is modeled by the unit in Figure~\ref{gtu:clique}.
As the number of unlabeled loops decreases whenever the first rule is applied and the number of \lp{\alpha}s decreases whenever the second rule is applied, one gets $\DEC(\clique) \in \NP$.
A successful derivation for a concrete instance is depicted in Figure~\ref{fig:clique-example-derivation}. The application of the second rule as long as possible is omitted because the number of applications is $0$ if the derivation is successful.
\end{example}

\section{Graph-Transformational Modeling of Reductions}
\label{sec:modeling-reductions}
Reductions between $\NP$-problems are a key concept for the study of the class $\NP$. Given two $\NP$-problems $\DEC$ and $\DEC'$, a reduction $\RED$ is a function from the inputs of $\DEC$ to the inputs of $\DEC'$ that can be computed in polynomial time and is subject to the following correctness condition for each input $x$ of~$\DEC$:
$\DEC(x) = \TRUE  \text{ if and only if } \DEC'(\RED(x)) = \TRUE.$
If the $\NP$-problems are given by polynomial graph transformation units, then a reduction can be modeled as a functional and polynomial graph transformation unit.

\begin{definition}[Reduction]
Let $\gtu$ and $\gtu'$ be polynomial graph transformation units.
Let $\red = (I_{\gtu},$ $P_{\red},C_{\red},I_{\gtu'})$ be a functional and polynomial graph transformation unit.
Then $\red$ is a \emph{reduction} of $\gtu$ to $\gtu'$ if the following holds for all $G \in \SEM(I_{\gtu})$:
$(G,H) \in \SEM(\gtu)$ for some $H \in \SEM(T_{\gtu})$ if and only if $(\red(G),H') \in \SEM(\gtu')$ for some $H' \in \SEM(T_{\gtu'})$.
\end{definition}

\begin{observe}
Let $\red$ be a reduction of $\gtu$ to $\gtu'$. Then the function $\RED = \SEM(\red)\colon \SEM(I_{\gtu}) \to \SEM(I_{\gtu'})$,
defined by $\RED(G) = \red(G)$,
is a reduction of $\DEC(\gtu)$ to $\DEC(\gtu')$, denoted by $\DEC(\gtu)$ $\leq$ $\DEC(\gtu')$.
\end{observe}

\begin{proof}
As $\red$ is functional and polynomial, $\RED$ is a function that is computed in polynomial time.
Its correctness can be shown as follows:
Let $G \in \SEM(I_{\gtu})$ with $\DEC(\gtu)(G) = \TRUE$ meaning that there is an $H \in \SEM(T_{\gtu})$ with $(G,H) \in \SEM(\gtu)$. As $\red$ is a reduction of $\gtu$ to $\gtu'$, this is equivalent to $(\red(G),H') \in \SEM(\gtu')$ for some $H' \in \SEM(T_{\gtu'})$, i.e., $\DEC(\gtu')(\red(G)) = \TRUE$.
As $\red(G) = \RED(G)$, $\RED$ turns out to be a reduction of $\DEC(\gtu)$ to $\DEC(\gtu')$.
\end{proof}

The observation means that $\DEC(\gtu) \leq \DEC(\gtu')$ can be established by means of a reduction of $\gtu$ to $\gtu'$.
The polynomiality of such a reduction can be shown by using Fact~\ref{fact_polynomial_termination}.
The functionality can be shown using the following well-known fact.

\begin{fact}
\rm
Let $\gtu = (I,P,C,T)$ be a polynomial graph transformation unit.
Then $\gtu$ is functional if the following hold:
\begin{enumerate}
\item $P$ is independent,
\item $(G,H) \in \SEM(\gtu)$ implies $H \in \mathit{reduced}(P)$, and
\item $G \dder^{*}_{P,C} H$ for $G \in \SEM(I)$ and $H \in \mathit{reduced}(P)$ implies $(G,H) \in \SEM(\gtu)$.
\end{enumerate}
\end{fact}

\begin{proof}
As $\gtu$ is polynomial, the lengths of its permitted derivations are finitely bounded.
The required independence implies the confluence of $P$-derivations.
Both together means that for each $G \in \SEM(I)$, there is -- up to isomorphism -- exactly one $H \in \mathit{reduced}(P)$ with $G \dder^{*}_{P} H$.
Properties 2 and 3 make sure that this functionality holds for $\SEM(\gtu)$, too.
\end{proof}

\begin{example}\label{example:reduction-from-hampath-to-stwbdk}
According to the considerations in Examples~\ref{example_hampath} and~\ref{example_stwbd}, we have $\DEC(\hampath)(G) = \TRUE$ if and only if $G$ has a Hamiltonian path, and we have $\DEC(\stwbd{k})(G) = \TRUE$ if and only if $G$ has a spanning tree with a vertex degree not greater than $k$.
A Hamiltonian path is connected and cycle-free and visits all vertices such that it is a spanning tree with vertex degree 2.
Conversely, a spanning tree with vertex degree 2 or smaller must be a simple path and, as it covers all vertices, therefore, a Hamiltonian path.
In other words, the identity on the initial graphs of $\hampath$ and $\stwbd{2}$ is a reduction of $\hampath$ to $\stwbd{2}$,
as specified in the unit in Figure~\ref{gtu:hampath-to-stwbdtwo}.

As there are no rules to be applied and as the regular expression $\epsilon$ specifies the language with the empty string as its sole element allowing derivations of length 0 only, each initial graph derives itself exclusively, which is terminal at the same time because the initial graphs of $\hampath$ and $\stwbd{2}$ coincide.
The unit is obviously polynomial and functional.
\end{example}

\begin{figure}[t]
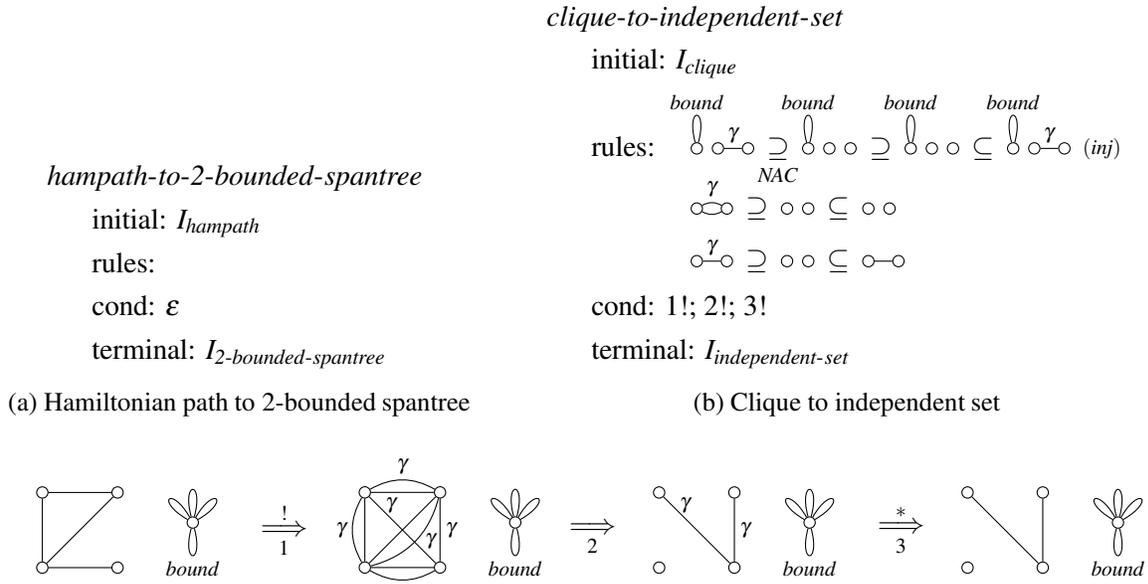

\centering
 \begin{subfigure}[b]{.4\textwidth}
		\centering
		\onehalfspacing
%\
\begin{tabular}{l}
$\hampathtostwbdtwo$\\
\quad
\begin{tabular}{l}
initial: $I_{\hampath}$\\
rules:\\
cond: $\epsilon$\\
terminal: $I_{\stwbd{2}}$
\end{tabular}
\end{tabular}
%\singlespacing
		\caption{Hamiltonian path to $2$-bounded spantree}
		\label{gtu:hampath-to-stwbdtwo}
	\end{subfigure}
	\hfill
	\begin{subfigure}[b]{.59\textwidth}
		\centering
		\input{units/cliquetoindependentset}
		\caption{Clique to independent set}
		\label{gtu:clique-to-independent-set}
	\end{subfigure}
	\begin{subfigure}[b]{\textwidth}
        \centering
        \vspace{1em}
        \input{tikz/example-cliquetoindependentset-1-derivation}
	\caption{Example derivation in $\cliquetoindependentset$}
	\label{fig:cliquetoindependentset-example-derivation}
        \end{subfigure}
	\caption{Graph transformation units for reductions}
	\label{gtu:reductions}
\end{figure}

\begin{example}\label{exp:reduction-clique-independentset}
The reduction from $\clique$ to $\independentset$ is modeled by the unit in Figure~\ref{gtu:clique-to-independent-set}.
Each initial graph has the form $G + \bound{k}$ for some $k \in \Nat$.
The first rule can be applied to two vertices of $G$ that are not connected by a $\gamma$-edge and neither of them has a $\mathit{bound}$-loop, connecting them by a $\gamma$-edge. %such an edge.
Hence the rule can be applied $\cardinality{V_G} \cdot (\cardinality{V_G} - 1)/2$ times and must be applied as often due to the control condition.
Then, any two vertices of $G$ are connected by $\gamma$-edges.
The second rule can be applied to an unlabeled edge of $G$ with a parallel $\gamma$-edge, removing both.
Due to the control condition, this must be done for all such pairs of edges in $G$.
The number of steps is $\cardinality{E_G}$, and $\gamma$-edges remain between all pairs of vertices of $G$ that are not connected in $G$.
By applying the third rule as long as possible, all $\gamma$-edges are replaced by (unlabeled) edges such that the resulting graph is isomorphic to the dual one of the initial graph.
In particular, the unit turns out to be functional and polynomial.
The overall lengths of derivations is $\cardinality{V_G} \cdot (\cardinality{V_G} - 1)$ for $G \in \SEM(I)$.
The rule $i$ for $i = 1,2,3$ is independent of itself, so that its application as long as possible is confluent in every case.
And as all intermediate graphs of derivations contain $\gamma$-edges, the terminal graphs are only reached at the end.
Finally, the unit has the correctness property of a reduction.
If two vertices in a graph $G$ are connected, then they are not connected in the dual graph and the other way round.
Consequently, a clique in $G$ becomes an independent set in the dual graph and the other way round.
A derivation for the initial graph of Figure~\ref{fig:clique-example-derivation} is depicted in Figure~\ref{fig:cliquetoindependentset-example-derivation}.
\end{example}

The correctness proof in Example~\ref{exp:reduction-clique-independentset} can be found in Karp~\cite{Karp:72} and in Garey and Johnson~\cite{Garey-Johnson:79}.
%This traditional graph theoretic way of proving the correctness of reductions can always be employed for proving the correctness of a graph-transformational reduction if the reduction as well as the two $\NP$-problems modeled by the two units are characterized in graph-theoretic terms.
Whenever a graph-transformational reduction and the two involed NP problems modeled by graph transformation
units are characterized in graph-theoretic terms, one can try to prove the correctness of the reduction in the traditional graph-theoretic way as in the
Examples 6 and 7.

In the next section, we investigate an alternative graph-transformational way of proving the correctness of reductions.

\section{Proving the Correctness of Reductions}
\label{sec:correctness-reductions}
In this section, we investigate possibilities to prove the correctness of reductions by graph-trans\-for\-ma\-tio\-nal means.

Let $\gtu$ and $\gtu'$ be two polynomial graph transformation units and $\red = (I_{\gtu}, P_{\red}, C_{\red}, I_{\gtu'})$ be a polynomial and functional graph transformation unit.
Then $\red$ is a reduction of $\DEC(\gtu)$ to $\DEC(\gtu')$ if the following correctness conditions hold:

\begin{paragraph}{Forward}
If $G \dder^{*}_{\gtu} H$ is a successful derivation of $\gtu$, then there is a successful derivation $\red(G) \dder^{*}_{\gtu'} H'$ for some $H' \in \SEM(T_{\gtu'})$.
\end{paragraph}

\begin{paragraph}{Backward}
If $\red(G) \dder^{*}_{\gtu'} H'$ is a successful derivation of $\gtu'$, then there is a successful derivation $G \dder^{*}_{\gtu} H$ for some $H \in \SEM(T_{\gtu})$.\\
\end{paragraph}

The diagram in Figure~\ref{fig:derivationstructures-schema} illustrates the task.
Given a $\gtu$- and a $\red$-derivation, one must find an appropriate $\gtu'$-derivation, and conversely given a $\gtu'$-derivation and the same $\red$-derivation, one must find an appropriate $\gtu$-derivation.
To achieve this, we propose a kind of toolbox that provides operations on derivation structures where the latter are defined in Definition~\ref{def:derivation-structure} and the operations on them in Definition~\ref{def:derivation-structure-operations}.

A derivation structure is a finite directed graph, the vertices of which are graphs and the edges are direct derivations.

\begin{definition}\label{def:derivation-structure}
%\begin{enumerate}
%\item
A \emph{derivation structure} is a finite unlabeled directed graph $\DS$, such that $V_{\DS} \subseteq \mathcal{G}_{\Sigma}$ and every edge $e \in E_{\DS}$ is a direct derivation $e = (G \dder_{r} H)$ with $s_{\DS}(e) = G$ and $t_{\DS}(e) = H$.
%\item
The class of all derivation structures is denoted by $\CDS$.
%\end{enumerate}
\end{definition}

The operations add derivations to given derivation structures so that the results are derivation structures.
The idea is
%\begin{enumerate}
%\item
(1)~to start with the derivation structure in Figure~\ref{fig:derivationstructures-a} and to apply operations until the derivation $\red(G) \dder^{*}_{\gtu'} H'$ appears as a substructure and
%\item
(2)~to start conversely with the derivation structure in Figure~\ref{fig:derivationstructures-b} and apply operations until the derivation $G \dder^{*}_{\gtu} H$ appears as a substructure.
%\end{enumerate}

\begin{figure}[t]
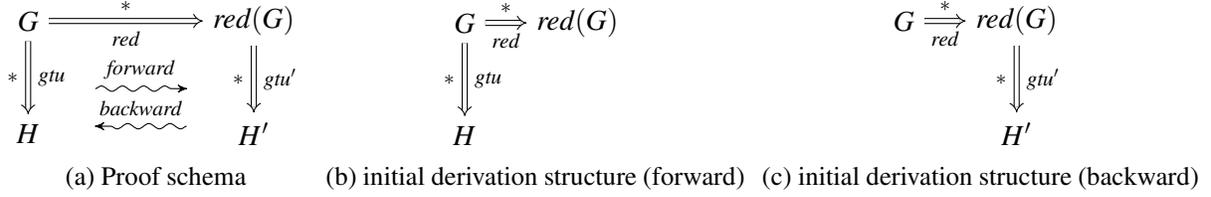

\begin{subfigure}[b]{0.26\textwidth}
\centering
\input{tikz/reduction}
\caption{Proof schema}
\label{fig:derivationstructures-schema}
\end{subfigure}
\hfill
\begin{subfigure}[b]{0.35\textwidth}
\centering
\input{tikz/derivationstructures-a}
\caption{initial derivation structure (forward)}
\label{fig:derivationstructures-a}
\end{subfigure}
\hfill
\begin{subfigure}[b]{0.37\textwidth}
\centering
\input{tikz/derivationstructures-b}
\caption{initial derivation structure (backward)}
\label{fig:derivationstructures-b}
\end{subfigure}
\caption{Correctness proof schema and initial derivation structures}
\label{fig:derivationstructures}
\end{figure}

We provide five operations which are defined as binary relations on derivation structures.
(1)~$\conflux$ of two parallel independent rule applications adding 
two corresponding rule applications (cf. Section~\ref{sec:preliminaries}) 
to the derivation structure;
(2)~$\interchange$ of two sequential independent rule applications adding one corresponding derivation consisting of the interchanged rule applications (cf. Section~\ref{sec:preliminaries}) to the derivation structure;
(3)~$\sprout$ with a functional graph transformation unit as parameter
attaching a successful derivation of the parameter starting in a graph of the given derivation structure;
(4)~$\couple$ with a span of parallel independent rule applications as parameter adding a corresponding extended right part of the span to the derivation structure if an extension of the left part of the span is present; and
(5)~$\associate$ with a pair of particular derivations with the same first graph and the same last graph as parameter adding an extension of the second derivation to the derivation structure if a corresponding extension of the first one is present.
All these enlargements of the derivation structures are formally defined by unions of derivation structures.
Concerning $\conflux$ and $\interchange$, the rule applications corresponding to the
given independent ones are unique up to isomorphism. One of the choices is
taken. Concerning $\sprout$, there are many derivations $G \dder^* \funct(G)$ with a
result unique up to isomorphism. Again one of them is chosen. Concerning
$\couple$ and $\associate$, the extensions are unique up to isomorphism provided they
exist at all. Then one of the extensions is added. The extensions may not exist
because the right component of the span or the pair may remove
vertices that are needed for the extensions. In such a case, the respective
operation is undefined.

\begin{definition}\label{def:derivation-structure-operations}
\begin{enumerate}

\item $\conflux$: Let $\DS \in \CDS$ with two parallel independent direct derivations $G \dder_{r_i} H_i$ for $i = 1,2$ as substructure.
Let $\DS'$ be the enlargement of $\DS$ by two corresponding direct derivations $d_1 = (H_1 \dder_{r_2} X)$ and $d_2 = (H_2 \dder_{r_1} X)$ provided by the parallel independence,
i.e., $\DS' = \DS \cup d_1 \cup d_2$.
Then $(\DS,\DS') \in \conflux$.

\item $\interchange$: Let $\DS \in \CDS$ with two sequential independent direct derivations $G \dder_{r_1} H_1 \dder_{r_2} X$ as substructure.
Let $\DS'$ be the enlargement of $\DS$ by the corresponding derivation $d =\linebreak (G \dder_{r_2} H_2 \dder_{r_1} X)$ provided by the sequential independence, i.e.,
$\DS' = \DS \cup d$.
Then $(\DS,\DS') \in \interchange$.

\item $\sprout(\funct)$: Let $\funct$ be a functional graph transformation unit, $\DS \in \CDS$ and $G \in V_{\DS}$.
Let $\DS'$ be the enlargement of $\DS$ by attaching one of the derivations $d = (G \dder^{*} \funct(G))$ at $G$, i.e.,
$\DS' = \DS \cup d$.
Then $(\DS,\DS') \in \sprout(\funct)$.

\item $\couple(\mathit{span})$: Let $\mathit{span}$ be a pair of parallel independent direct derivations $d_i = (G \dder_{r_i} H_i)$ for $i = 1,2$, $\DS \in \CDS$, and $\widehat{G} \dder \widehat{H}_1$ be a substructure of $\DS$ extending $d_1$.
Let $\DS'$ be the enlargement of $\DS$ by adding one corresponding extension $\widehat{d}_2 = (\widehat{G} \dder_{r_2} \widehat{H}_2)$ of $d_2$, i.e.,
$\DS' = \DS \cup \widehat{d}_2$, provided it exists.
Then $(\DS,\DS') \in \couple(\mathit{span})$.

\item $\associate(\mathit{pair})$: Let $\mathit{pair}$ be a pair of derivations $d_i = (G \dder^* H)$ for $i = 1,2$, $\DS \in \CDS$, $\widehat{d}_1 = (\widehat{G} \dder^{*} \widehat{H})$ be a substructure of $\DS$ extending $d_1$.
Let $\DS'$ be the enlargement of $\DS$ by adding one extension $\widehat{d}_2$ of $d_2$ from $\widehat{G}$ to $\widehat{H}$, i.e.,
$\DS' = \DS \cup \widehat{d}_2$, provided it exists.
Then $(\DS,\DS') \in \associate(\mathit{pair})$.

\end{enumerate}
\end{definition}

Figure~\ref{fig:schematic-drawings-operations} depicts schematic drawings of the operations.
\begin{figure}[th!]
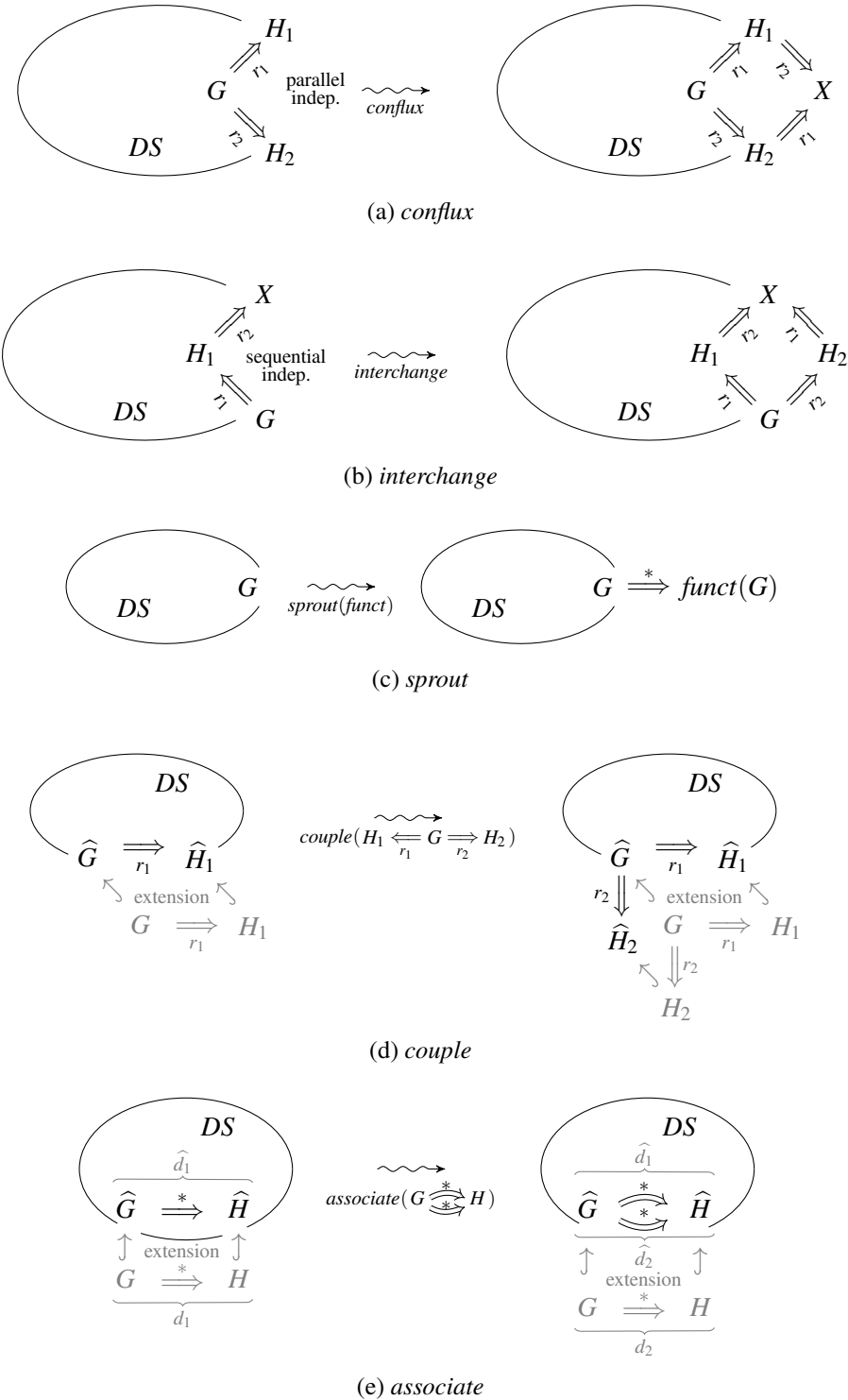

\begin{subfigure}[b]{\textwidth}
\centering
\input{tikz/conflux}
\caption{$\conflux$}
\label{fig:schema-conflux}
\end{subfigure}
\hfill
\begin{subfigure}[b]{\textwidth}
\centering
\vspace{1em}
\input{tikz/interchange}
\caption{$\interchange$}
\label{fig:schema-interchange}
\end{subfigure}
\begin{subfigure}[b]{\textwidth}
\centering
\vspace{1em}
\input{tikz/sprout}
\caption{$\sprout$}
\label{fig:schema-sprout}
\end{subfigure}
\begin{subfigure}[b]{\textwidth}
\centering
\vspace{2em}
\input{tikz/couple}
\caption{$\couple$}
\label{fig:schema-couple}
\end{subfigure}
\begin{subfigure}[b]{\textwidth}
\centering
\vspace{1em}
\input{tikz/associate}
\caption{$\associate$}
\label{fig:schema-associate}
\end{subfigure}
\caption{Schematic drawings of the operations}
\label{fig:schematic-drawings-operations}
\end{figure}

In the following examples, we demonstrate how the operations may be employed to prove the correctness of reductions.

\begin{example}
\label{example:correctness-clique-to-independent-set}
Using the operations $\conflux$, $\interchange$ and $\sprout$, the correctness of\linebreak $\cliquetoindependentset$ can be shown as follows.
Let $G \dder^{!}_{1} G_1 \dder^{!}_{2} G_2$ be a successful derivation in $\clique$.
As $G_2$ is terminal, it has no $\beta$-edges so that $G_1 \dder^{!}_{2} G_2$ has length $0$ and can be omitted.
Let $G \dder^{!}_{\overline{1}} \overline{G_1} \dder^{!}_{\overline{2}} \overline{G_2} \dder^{!}_{\overline{3}} \overline{G_3}$ be a successful derivation of $\cliquetoindependentset$
reducing $G$ where the numbers of rules are overlined to distinguish them from the rules $1$ and $2$ of $\clique$.
Each rule application of \rl{1} is parallel independent of each rule application of rule $\overline{i}$ for $i = 1,2,3$.
Therefore, the application of $\conflux$ as long as possible to the rule applications of the two given derivations yields - among others - a derivation $\overline{G_3} \dder^{!}_{1} G'_{1}$ being a moved variant of $G \dder^{!}_{1} G_1$ as defined at the end of Section~\ref{sec:preliminaries}.
Figure~\ref{fig:derivationstructure1} depicts the derivation structure.
\begin{figure}[t]
\begin{subfigure}[b]{\textwidth}
\centering
\input{tikz/derivation-grid-example1}
\caption{Abstract derivation structure}
\label{fig:derivationstructure1}
\end{subfigure}
\begin{subfigure}[b]{\textwidth}
\centering
\vspace{1em}
\input{tikz/example-cliquetoindependentset-grid.tex}
\caption{Derivation structure for forward and backward correctness using the concrete examples presented above}
\label{fig:derivations-clique-proof-procedure-example}
\end{subfigure}
\caption{Abstract and concrete derivation structures in Example~\ref{example:correctness-clique-to-independent-set}}
\end{figure}

The involved rule applications keep the set of vertices invariant.
Moreover, the independence of all rule applications makes sure that the sets of vertices with \lp{\alpha}s in $G_1$ and $G'_1$ are equal.
Let $v$ and $v'$ be two such vertices which are connected by an edge $e$ in $G$ according to the choice of $G \dder^{!}_{1} G_1$.
Consequently, $e$ gets a parallel $\gamma$-edge in $\overline{G_1}$, both are removed between $\overline{G_1}$ and $\overline{G_2}$, and $v$ and $v'$ are not connected in $\overline{G_3}$.
This means that $G'_1$ is a terminal graph of $\independentset$.
Altogether, the reduction is forward correct.
The backwards correctness follows analogously, starting from the derivation reducing $G$ and a successful derivation $\overline{G_3} \dder^{!}_{1} G'_1$ in $\independentset$ and applying $\interchange$ (instead of $\conflux$) to the sequential independent rule applications of \rl{\overline{i}} for $i = 1,2,3$, followed by one application of \rl{1}.
This yields a derivation $G \dder^{!}_{1} G_1$ as a moved variant of $\overline{G_3} \dder^{!}_{1} G'_1$.
As each two vertices with \lp{\alpha}s in $G'_1$ are not connected, the definition of the reduction implies that they are connected in $G_1$.
In particular, \rl{2} of $\clique$ cannot be applied and the derivation $G_1 \dder^{!}_{2} G_2$ has length $0$.
Adding it with $\sprout$ at $G_1$, one gets a successful derivation in $\clique$.
This completes the proof.

Figure~\ref{fig:derivations-clique-proof-procedure-example} depicts the derivation structure in Example~\ref{example:correctness-clique-to-independent-set} using the concrete examples used before.
\end{example}

\begin{figure}[h!]
\begin{subfigure}[b]{\textwidth}
\centering
\input{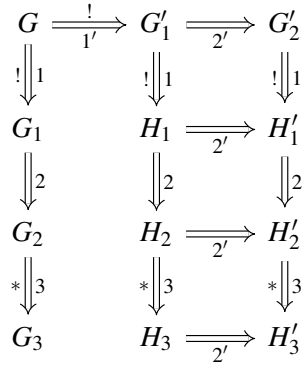}
\caption{Partial abstract derivation structure}
\label{fig:derivationstructure2}
\end{subfigure}
\begin{subfigure}[b]{\textwidth}
\centering
\vspace{3em}
\input{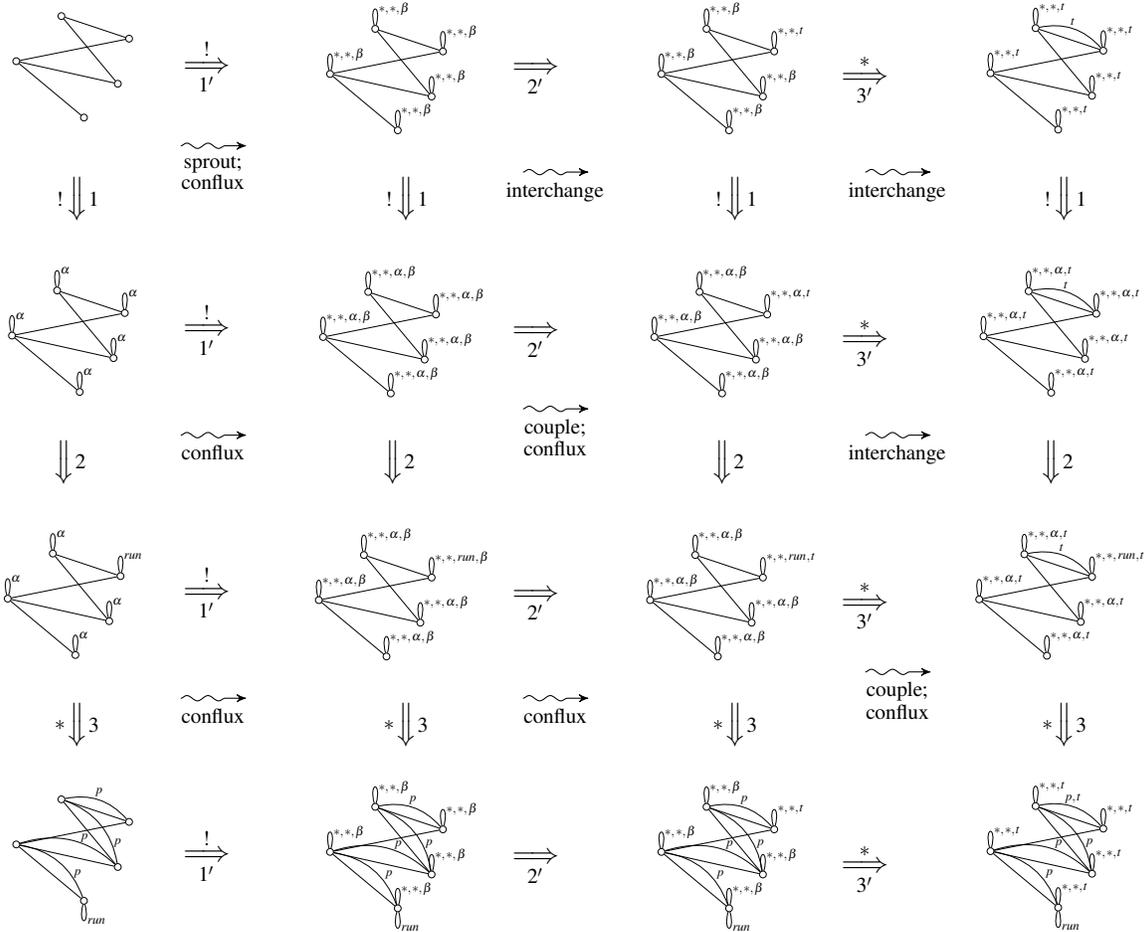}
\caption{Partial derivation structure for forward correctness using the concrete examples presented above}
\label{fig:derivations-hampathtostwdb2-proof-procedure-example}
\end{subfigure}
\caption{Abstract and concrete derivation structures in Example~\ref{example:correctness-hampath-to-stwbdtwo}}
\end{figure}

\begin{example}
\label{example:correctness-hampath-to-stwbdtwo}
In proving the correctness of $\hampathtostwbdtwo$, all five operations are used.
Let $G \dder^{!}_1 G_1 \dder_{2} G_2 \dder^{*}_{3} G_3$ be a successful derivation of $\hampath$.
This is a start structure as the reduction is the identity and one can begin to construct the corresponding derivation immediately.
Its first section is $G \dder^{!}_{1'} G'_{1}$ necessarily which can be added with $\sprout$ at $G$ using the functional unit with the rule $1'$ and $1'!$ as control condition.
The rules of $\stwbd{2}$ are primed to distinguish them from the rules of $\hampath$.
The rule sets $\{1,2,3\}$ and $\{1',2',3'\}$ are independent.
Therefore, $\conflux$ applied as long as possible yields the subderivation $G \dder^{!}_{1'} G'_1 \dder^{!}_{1} H_1 \dder_{2} H_2 \dder^{*}_{3} H_3$ where the last three sections are a moved variant of the given derivation so that $H_1$ contains $G$ and an \lp{\alpha}, a \lp{\beta} and two \lp{*}s at each vertex and $H_3$ contains $G_3$ and a \lp{\beta} and two \lp{*}s at each vertex.
In particular, $H_1 \dder_{2} H_2$ is an extension of the left part of the span
$\input{tikz/example2-span1}$.
Therefore, $\couple$ can be applied adding the direct derivation $H_1 \dder_{2'} H'_1$.
Using the known independence, the latter can be interchanged with the rule applications of $G'_1 \dder^{!}_{1} H_1$ yielding, in particular, $G'_1 \dder_{2'} G'_2$ and $G'_2 \dder^{!}_{1} H'_1$.
Moreover, $\conflux$ can be applied to $H_1 \dder_{2'} H'_1$ and the rule applications of $H_1 \dder_{2} H_2 \dder^{*}_{3} H_3$ yielding the moved variant $H'_1 \dder_{2} H'_2 \dder^{*}_{3} H'_3$.
Figure~\ref{fig:derivationstructure2} shows the main parts of the derivation structure that is constructed so far.
If $H'_2 \dder^{*}_{3} H'_3$ has length $0$, then we are done.
The uppermost horizontal derivation is the result we are looking for.
Otherwise, one must repeat the construction done for $H_1 \dder_{2} H_2$, i.e. couple, interchange, and then conflux as long as possible, for each of the applications of \rl{3} in $H'_2 \dder^{*}_{3} H'_3$ one after the other using the span
$\input{tikz/example2-span2}$.
In the end, this yields a derivation $G'_2 \dder^{*}_{3'} G'_3$ with the same length as $G_2 \dder^{*}_{3} G_3$, i.e., $n-1$ if $G$ has $n$ vertices.
Summarizing, the resulting derivation $G \dder^{!}_{1'} G'_{1} \dder_{2'}$ $G'_2 \dder^{*}_{3'} G'_3$ is successful as the derivation picks $n-1$ edges which always is a spanning tree, therefore, the reduction turns out to be forward correct.

Conversely, let $G \dder^{!}_{1'} G'_1 \dder_{2'} G'_2 \dder^{*}_{3'} G'_3$ be a successful derivation of $\stwbd{2}$. Let $v_0$ be the vertex matched by rule $2'$ and $e_1,\ldots,e_{n-1}$ be the edges matched by the following applications of rule $3'$ in this order, where $e_1$ is attached to $v_0$.
If none of the further edges is attached to $v_0$, then the arguments proving forward correctness can be converted so that one gets backwards correctness, too.
Otherwise, let $e_i$ with $i \geq 2$ be attached to $v_0$.
Without loss of generality, one can assume $i = 2$ because this case can always be obtained by interchanges.
Then the derivation $G'_1 \dder_{2'} G'_2 \dder_{3'} \overline{G} \dder_{3'} \overline{\overline{G}}$ is an extension of the upper derivation in the pair 
%\vspace{-.1cm}
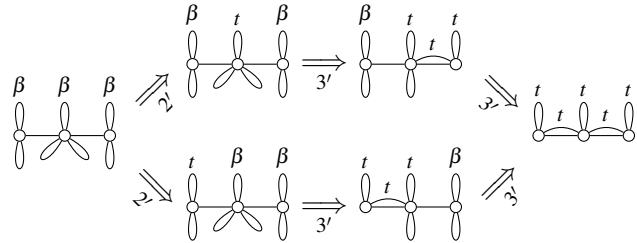
\[\input{tikz/example2-association}\]
so that $\associate$ can be applied.
This means that each successful derivation can be rearranged by repeated interchanges and associations in such a way that the conversion of the forward correctness proof works.

Figure~\ref{fig:derivations-hampathtostwdb2-proof-procedure-example} depicts the derivation structure in Example~\ref{example:correctness-hampath-to-stwbdtwo} using the concrete examples used before.

\end{example}

There are several further reductions the correctness of which can be shown in a similar line of consideration, but their documentation here is beyond the space limit of the paper. 
Examples are $\independentset \leq \clique$, $\stwbd{k} \leq \stwbd{l}$ for $2 \leq k < l$ and $\hampath$ $\leq$\linebreak$\hamcycle$ $\leq$ $\TSP$ where $\TSP$ is the famous traveling salesperson problem. 
The latter example works very similar to $\clique \leq \independentset$ with the exception that the successful derivations of $\hamcycle$ are not moved to the end of the reduction derivation, but only to some intermediate graph where the corresponding $\TSP$-derivation is constructed, which is further moved along the reduction.
The examples indicate that the following proof procedure is quite promising to result in the correctness of reductions between $\NP$-problems.

\begin{proofprocedure}
\rm
Let $\gtu$ and $\gtu'$ be two polynomial graph transformation units and $\red = (I_{\gtu}, \overline{P_1} \cup \overline{P_2}, \overline{C}, I_{\gtu'})$ with $\overline{P_1} \cap \overline{P_2} = \emptyset$ be a functional graph transformation unit such that each of its successful derivations has the form $G \dder_{\overline{P_1}}^{*} \overline{G} \dder_{\overline{P_2}}^{*} \red(G)$.
Let $P_{\gtu}$ and $\overline{P_1}$, $\overline{P_2}$ and $P_{\gtu'}$, as well as $P_{\gtu}$ and $P_{\gtu'}$ be disjoint and independent.
Then the correctness of $\red$ may be proved as follows:

%\begin{enumerate}[wide,labelindent=0pt,topsep=0pt, partopsep=0pt]
\begin{paragraph}{Forward}
Consider the start structure of Figure~\ref{fig:derivationstructures-a}.
%\begin{enumerate}[wide,leftmargin=20pt,labelindent=20pt,topsep=0pt, partopsep=0pt]
\begin{itemize}[leftmargin=21pt]
\item[(f1)] Apply $\conflux$ as long as possible using the independence of $P_{\gtu}$ and~$\overline{P_1}$.
The moved variant of the given $\gtu$-derivation starting at $\overline{G}$ is called \emph{guide} and $\overline{G} = \AS$ \emph{active spot} being the start graph of a permitted $\gtu'$-derivation that is built further on.
\item[(f2)] Repeat the following operations as long as possible depending on the next rule application permitted by the stepwise control condition $C_{\gtu'}$.
%\begin{enumerate}[wide,leftmargin=20pt,labelindent=20pt,topsep=0pt, partopsep=0pt]
\begin{itemize}
\item[(f21)] If the next rule to be applied is the first one of a functional section of a permitted derivation given by $\funct$, then apply $\sprout(\funct)$ to $\AS$ where the graph $\funct(\AS)$ is the new $\AS$.
Then apply $\conflux$ as long as possible using the independence of $P_{\gtu}$ and $P_{\gtu'}$.
The moved variant of the guide is the new guide.
\item[(f22)] Otherwise, apply $\couple(\mathit{span})$ for some given $\mathit{span}$ to the first direct derivation of the guide where this is possible, provided that the prolongation yields a permitted derivation.
Then apply $\conflux$ and $\interchange$ as long as possible.
In particular, the added direct derivation is moved to $\AS$ prolonging the permitted $\gtu'$-derivation.
The result graph is the new $\AS$ and the following $\gtu$-derivation is the new guide.
\end{itemize}
\item[(f3)] Afterwards, apply $\conflux$ as long as possible using the independence of $\overline{P_2}$ and $P_{\gtu'}$ moving the constructed permitted $\gtu'$-derivation to $\red(G)$.
Then continue the derivation if $C_{\gtu'}$ requires further rule application.
\item[(f4)] Check whether the derived graph is in $\SEM(T_{\gtu'})$.
\end{itemize}
\end{paragraph}
\begin{paragraph}{Backward}
Reverse the procedure of the forward proof by starting with the structure of Figure~\ref{fig:derivationstructures-b} and changing the roles of $\gtu$ and $\gtu'$.
%\begin{enumerate}[wide,leftmargin=20pt,labelindent=20pt,topsep=0pt, partopsep=0pt]
\begin{itemize}[leftmargin=23pt]
\item[(b1)] Apply $\interchange$ instead of $\conflux$ as long as possible using the independence of $\overline{P_2}$ and $P_{\gtu'}$.
\item[(b2)] Repeat the operation as in (f2), but now using the spans conversely.
\item[(b3)] Afterwards, apply $\interchange$ as long as possible using the independence of $P_{\gtu}$ and $\overline{P_1}$.
\item[(b4)] Check whether the derived graph is in $\SEM(T_{\gtu})$.
\end{itemize}
\end{paragraph}
Figure~\ref{fig:derivation-proof-prof-forward-correctness} depicts the derivations for proving correctness.
\begin{figure}[h!]
\centering
\input{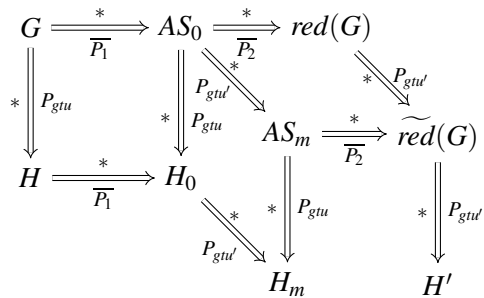}
\caption{Derivations for proving correctness}
\label{fig:derivation-proof-prof-forward-correctness}
\end{figure}
\begin{paragraph}{Preprocessing}
As the backward proof of $\hampathtostwbdtwo$ shows, a preprocessing may be necessary before the backward part of the proof procedure can work.
The reason is that a graph may have several successful derivations, but only some of them may be suitable for forward and backward processing.
In such a case, one may apply $\associate(\mathit{pair})$ for appropriate pairs of derivations together with $\interchange$.
\end{paragraph}
\end{proofprocedure}

It may be noted that our proof for $\clique \leq \independentset$ in Example~\ref{example:correctness-clique-to-independent-set} must be modified to be covered by the proof procedure in the following way: 
$\alpha$ is replaced by $\alpha'$ in $\independentset$ and the corresponding derivation is constructed by applying $\couple$ for each rule application where the span $\input{tikz/independentset-couple-span}$ is used.

The proof procedure works very well for several examples, but is far from
perfect. For instance, we have failed to prove the correctness of the
reduction from independent set to vertex cover and from vertex cover to
Hamiltonian cycles. The reason seems to be that one needs further operations.
But there are various further reasons why the proof procedure may fail:
(1)~The chosen units for the involved NP problems and the reduction may not be suitable.
(2)~The preprocessing may not yield a successful derivation for further processing.
(3)~The processing may get stuck because the proper spans are missing.
(4)~The processing may result in a permitted derivation that is not successful.
Such problems are not unusual if the employed proof technique is not complete
as it is often the case. Fortunately, each of the problems indicate how one
may find a way out: Try further operations, other modeling units, other spans
for coupling, other pairs for association, new ideas for preprocessing. See
the conclusion where some of these points are a bit more elaborated.

\section{Conclusion}
\label{sec:conclusion}
We have made a proposal on how the correctness of reductions between $\NP$ problems may be proved by graph-transformational means. 
For this purpose, we have provided a toolbox that allows to construct successful derivations of one graph transformation unit from successful derivations of another graph transformation unit that represent positive decisions and are connected by a reduction derivation in their initial graphs. 
The approach is an attempt in the early stage of development and needs further investigation, to shed more light on its usefulness.
\vspace{-.2em}
\begin{enumerate}
\itemsep.1em
\item
The proof procedure has turned out to be suitable in several cases, but it seems to fail in other cases like, for example, for the sophisticated reduction from the vertex cover problem to the Hamiltonian cycle problem (cf.~\cite{Karp:72,Garey-Johnson:79}).
\item
To cover more cases or to simplify proofs, one may look for further operations. Candidates are operations like $\conflux$ and $\interchange$ that are based on independence including parallelization, sequentialization, and shift (cf.~\cite{kreowski1977transformations}).
\item
Another possibility is to generalize the coupling by considering spans that consist of derivations rather than direct derivations.
\item
The preprocessing needs more attention. An initial graph may have a large class of successful derivations of which only a few particular ones may be suitable for the further proof procedure. Therefore, operations that preserve successfulness are of interest in addition to the association.
\item
The toolbox relies heavily on independence. Can this be always achieved by respective units? Can it be relaxed?
\item
Is there any chance of tool support for such correctness proofs? For particular initial graphs, the proof procedure is of an algorithmic nature. But what about arbitrary successful derivations and reduction derivations?
\item
A reduction between $\NP$ problems is a kind of model transformation. Do results of the theory of model transformations and their correctness proofs like triple graph grammars (cf., e.g.,~\cite{schurr200815}) apply to reductions?
\end{enumerate}

\vspace{-.9em}
\paragraph{Acknowledgment.}
We are grateful to the anonymous reviewers for their valuable comments.

\vspace{-.9em}
\bibliographystyle{eptcs}
\bibliography{lit,lit2,lit_all}

\end{document}